\documentclass[11pt,letterpaper]{article}

\usepackage[utf8]{inputenc}
\usepackage[T1]{fontenc}
\usepackage{mathptmx}
\usepackage{amsmath,amssymb,amsthm}
\usepackage{graphicx}
\usepackage{booktabs}
\usepackage{natbib}
\usepackage[margin=1in]{geometry}
\usepackage{hyperref}
\usepackage{xcolor}
\usepackage{float}
\usepackage{subcaption}
\usepackage{setspace}
\usepackage{algorithm}
\usepackage{algpseudocode}

\hypersetup{
    colorlinks=true,
    linkcolor=blue!60!black,
    citecolor=blue!60!black,
    urlcolor=blue!60!black
}

\newtheorem{theorem}{Theorem}

\theoremstyle{definition}

\newcommand{\E}{\mathbb{E}}

\title{\Large\textbf{Equity Implications of Federal--Local Cost-Sharing in Flood Buyouts: A Game-Theoretic Analysis with Heterogeneous Homeowners}}

\author{
Yuqun Zhou\\[0.5em]
\small Department of Industrial Engineering, University of Wisconsin--Madison, Madison, WI, USA\\
\small Email: yzhou364@wisc.edu
}

\date{\today}

\begin{document}

\maketitle


\begin{abstract}
\noindent
Climate-driven flood risk increasingly necessitates managed retreat through government buyout programs. Yet empirical evidence documents substantial racial and economic disparities in program implementation. Here we develop a three-level Stackelberg game to analyze how federal--local cost-sharing arrangements generate inequitable outcomes through strategic interactions among federal authorities, local governments, and heterogeneous homeowners. Our model reveals three distinct mechanisms driving inequity: differential discount rates across income groups, local governments' tax-base preservation incentives, and participation thresholds that exclude fiscally constrained communities. Numerical analysis of 34,493 households across nine flood-prone US regions demonstrates that the current Federal Emergency Management Agency 75/25 cost-sharing arrangement produces a relocation ratio gap of 0.26---low-income households relocate at roughly one-quarter the rate of high-income households. Achieving near-equity requires federal cost shares of at least 85\%, though equity-weighted mechanisms can attain similar outcomes at 25\% lower cost. These findings provide a theoretical foundation for understanding observed disparities and identify policy levers for more equitable climate adaptation.
\end{abstract}

\vspace{1em}


\section{Introduction}

Managed retreat through property buyouts has emerged as a critical climate adaptation strategy as flood risk intensifies across the United States~\cite{siders2019strategic, mach2019managed}. The Federal Emergency Management Agency (FEMA) administers primary buyout programs through the Hazard Mitigation Grant Program (HMGP) and Building Resilient Infrastructure and Communities (BRIC) initiative, which typically operate on a 75/25 federal--local cost-sharing basis. However, mounting empirical evidence reveals troubling disparities: Elliott et al.~\cite{elliott2020racial} documented that buyout funding disproportionately flows to wealthier, white communities, while Greer and Binder~\cite{greer2017political} found systematic barriers facing low-income participants.

These disparities likely emerge from incentive structures embedded within cost-sharing arrangements, yet existing theoretical frameworks inadequately capture the multi-level strategic dynamics at play. The foundational game-theoretic model of Bier et al.~\cite{bier2019game} treats government as a unified decision-maker, abstracting away from the divergent objectives between federal and local authorities that may drive observed inequities. Subsequent work has extended these foundations: Zhou~\cite{zhou2022predisaster} developed agent-based models to capture resident heterogeneity and network effects in relocation decisions, while Zhou and Cen~\cite{zhou2025stochastic} integrated stochastic programming with machine learning to enhance predictive accuracy for subsidy impacts. Local governments facing a 25\% cost share confront competing pressures: reducing disaster exposure through buyouts versus preserving property tax revenue and maintaining community stability.

We address this gap by developing a three-level stochastic Stackelberg game that explicitly models strategic interactions among federal government, local governments, and heterogeneous homeowners under climate uncertainty. The federal government acts as Stackelberg leader, setting cost-sharing parameters that shape local government participation decisions and, ultimately, which households receive buyout offers. Our framework identifies three distinct channels through which current policy generates inequitable outcomes and evaluates alternative mechanisms for improving distributional outcomes.

\section{Results}

\subsection{Theoretical Framework and Equilibrium Characterization}

Our model captures a sequential game played under climate uncertainty. The federal government first announces a cost-sharing ratio $\alpha \in [0,1]$ specifying its share of buyout costs and a maximum subsidy cap $\bar{S}$. Local governments then simultaneously decide whether to participate and what subsidy levels to offer. Finally, homeowners choose optimal relocation timing given offered subsidies.

We characterize homeowners by type $\theta_h = (V_h, r_h, \ell_h)$ representing house value, discount rate, and location. Following Warner and Pleeter~\cite{warner2001personal} and Bier et al.~\cite{bier2019game}, we model low-income households as having higher effective discount rates ($r^{\text{high}} = 18\%$) than high-income households ($r^{\text{low}} = 12\%$) due to credit constraints. This assumption, while necessarily simplifying the heterogeneous financial circumstances of low-income households, captures the well-documented empirical regularity that credit-constrained individuals exhibit higher implicit discount rates in revealed-preference studies~\cite{hausman1979individual, frederick2002time}. The magnitude of this differential proves central to understanding inequitable outcomes, and we examine sensitivity to alternative specifications in our robustness analysis.

Each homeowner's optimal relocation timing balances discounted relocation costs against cumulative expected flood damages. For household $h$ receiving subsidy $S$, the optimal relocation year satisfies a threshold condition where the present value of future flood damages equals net relocation costs (see Section~\ref{sec:methods}). Higher subsidies induce earlier relocation, but the magnitude of this effect depends critically on discount rates: a uniform subsidy shifts low-discount-rate households' decisions more than high-discount-rate households'.

Local governments minimize expected costs comprising their buyout cost share, administrative expenses, disaster-related expenditures, and a tax-base preservation term weighted by parameter $\lambda_j$. This specification captures one dimension of political economy---fiscal pressures to maintain property tax revenue---while abstracting from other relevant factors including electoral cycles, constituent preferences, and administrative capacity constraints. The first-order condition for optimal subsidy $S_j^*$ reveals that higher federal cost shares $\alpha$ induce higher local subsidy offers, while stronger tax-base concerns reduce subsidies overall (see Appendix~\ref{app:proofs}).

We characterize subgame perfect Nash equilibrium through backward induction. Equilibrium existence follows from standard continuity arguments (Appendix~\ref{app:proofs}). The federal government's optimal policy trades off aggregate social costs against equity constraints, with the equity constraint binding for sufficiently stringent equity targets.

\subsection{Mechanisms Generating Inequity}

Our theoretical analysis identifies three distinct channels through which current cost-sharing arrangements generate inequitable outcomes (Fig.~\ref{fig:mechanisms}).

Differential subsidy effectiveness across income groups drives the \textit{discount rate channel}. Because low-income households discount future flood damages more heavily, they require larger subsidies to induce equivalent relocation timing. Specifically, to induce relocation in year $k$ for a household with discount rate $r$, the required subsidy equals relocation costs minus discounted future damages---a quantity that increases with $r$. Under uniform subsidies, low-income households therefore relocate later or not at all.

Local governments' fiscal incentives create the \textit{tax-base preservation channel}. When $\lambda_j > 0$, local governments prefer buyouts that minimize tax-base erosion. This creates complex targeting dynamics: local officials may favor lower-value properties to preserve aggregate tax revenue, but low-income households do not necessarily occupy lower-value properties, particularly in gentrifying neighborhoods or when lower-income families reside in inherited homes.

Variation in local governments' ability to meet cost-sharing requirements generates the \textit{participation channel}. The critical threshold $\alpha_j^{\text{crit}}$ above which jurisdiction $j$ participates depends inversely on fiscal capacity and directly on administrative costs. Communities with limited fiscal resources---often lower-income communities---require higher federal cost shares to participate. Below their threshold, these communities receive no buyout program at all.

\subsection{Numerical Analysis}

We apply our framework to a multi-region simulation encompassing 34,493 households across nine flood-prone US regions: four New York City areas (Brooklyn, Staten Island, Queens, Lower Manhattan), Houston, New Orleans, Miami-Dade, Charleston, and Norfolk. Property values range from \$245,000 (Norfolk) to \$1.85 million (Lower Manhattan), low-income population shares range from 15\% (Lower Manhattan) to 62\% (New Orleans), and 100-year floodplain exposure ranges from 35\% (Brooklyn) to 75\% (New Orleans). This heterogeneity allows us to examine how regional characteristics interact with federal policy to shape distributional outcomes.

Figure~\ref{fig:equity_alpha} presents our central finding: the relationship between federal cost-share ratio and the relocation ratio gap (RRG), defined as the ratio of low-income to high-income relocation rates. Under the current FEMA 75/25 arrangement ($\alpha = 0.75$), we observe RRG = 0.26, indicating that low-income households relocate at roughly one-quarter the rate of high-income households. At this policy, only 191 of 17,123 low-income households (1.1\%) participate in buyouts, compared with 693 of 17,370 high-income households (4.0\%).

Achieving near-equity (RRG $\geq$ 0.70) requires federal cost shares of at least 85\%, which increases federal expenditure from \$82 million to \$375 million---a 4.6-fold increase. Beyond $\alpha = 0.90$, diminishing returns set in: RRG increases from 0.87 to 0.93 as $\alpha$ rises from 0.90 to 1.00, while costs increase by an additional \$176 million. The transition from inequitable to near-equitable outcomes thus occurs over a relatively narrow policy range.

Regional heterogeneity substantially shapes these aggregate patterns (Fig.~\ref{fig:regional}). Lower Manhattan achieves high participation rates at lower $\alpha$ values due to high property values and severe flood exposure, while lower-income regions like Norfolk and New Orleans require $\alpha \geq 0.85$ for meaningful participation. This variation reflects both the tax-base channel (Lower Manhattan's high property values create strong preservation incentives) and the participation channel (Norfolk's limited fiscal capacity raises its participation threshold).

\subsection{Sensitivity Analysis}

We assess robustness to key parameter assumptions through systematic sensitivity analyses. The discount rate differential between income groups proves quantitatively important: under low heterogeneity ($r^{\text{high}} = 14\%$ versus $r^{\text{low}} = 12\%$), achieving RRG $\geq$ 0.70 requires only $\alpha = 0.78$, while extreme credit constraints ($r^{\text{high}} = 25\%$ versus $r^{\text{low}} = 8\%$) require $\alpha \geq 0.92$. This sensitivity underscores the importance of accurately characterizing household discount rates in policy design.

Subsidy cap variations interact with cost-sharing ratios. At the baseline cap of \$250,000, increasing $\alpha$ from 0.75 to 0.90 improves RRG from 0.26 to 0.87. At a higher cap of \$400,000, the same $\alpha$ increase yields RRG improvements from 0.31 to 0.91, suggesting complementarity between the two policy instruments.

Climate scenario uncertainty does not qualitatively alter our conclusions. Under pessimistic scenarios (1.0m sea-level rise by 2100), higher flood damages make buyouts more attractive for all households, but relative inequity persists. The RRG under current policy ranges from 0.24 (optimistic scenario) to 0.29 (pessimistic scenario), a modest variation that does not change the fundamental policy implications.

\subsection{Alternative Policy Mechanisms}

We evaluate three alternative policy designs against the current uniform cost-sharing arrangement (Fig.~\ref{fig:policy}).

\textit{Equity-weighted cost sharing} sets jurisdiction-specific federal shares according to $\alpha_j = \bar{\alpha} + \gamma(MHI_{\text{nat}} - MHI_j)$, where $MHI_j$ denotes median household income in jurisdiction $j$ and $\gamma > 0$ is a progressivity parameter. This mechanism directs higher federal support to lower-income communities, addressing the participation channel directly. With $\gamma = 0.10$, equity-weighted sharing achieves RRG = 0.78 at federal cost of \$420 million---25\% lower than the \$648 million required under uniform $\alpha = 0.90$ to achieve similar equity outcomes.

\textit{Income-tiered subsidies} offer supplemental subsidies $\Delta S$ to low-income households to compensate for the discount rate gap. The optimal tier differential that equalizes relocation rates across income groups depends on expected future damages and the discount rate differential. With $\Delta S = \$75,000$, tiered subsidies achieve RRG = 0.82 at moderate cost increase, though implementation requires means-testing that may create administrative burdens.

\textit{Minimum service requirements} mandate that participating jurisdictions achieve low-income relocation rates at least $\rho$ times high-income rates within their boundaries. This approach achieves RRG $\geq \rho$ by construction but may reduce overall program participation if $\rho$ is too stringent. At $\rho = 0.80$, we observe modest participation reduction (from 12 to 10 jurisdictions) while substantially improving equity within participating communities.

\subsection{Model Validation}

We validate model predictions against observed outcomes from four major historical buyout programs (Table~\ref{tab:validation}). The Staten Island post-Sandy program achieved 8.2\% relocation rates among eligible households; our model predicts 9.1\% (+11\% error). Harris County's post-Harvey program achieved 4.5\% relocation; we predict 5.2\% (+16\% error). New Jersey's proactive Blue Acres program achieved 6.8\% relocation; we predict 5.9\% ($-13$\% error). Charlotte-Mecklenburg's program achieved 12.1\% relocation; we predict 10.8\% ($-11$\% error).

These prediction errors of 11--16\% suggest reasonable model calibration while highlighting areas for refinement. The consistent pattern---slight overprediction for post-disaster programs and underprediction for proactive programs---suggests our model may underweight the urgency effect of recent flood experience.

\section{Discussion}

Our analysis reveals that current federal--local cost-sharing arrangements in flood buyout programs generate substantial inequity through three reinforcing mechanisms. The discount rate channel reflects fundamental differences in how credit-constrained households value present versus future outcomes. The tax-base preservation channel emerges from local political economy. The participation channel excludes fiscally constrained communities entirely. These mechanisms operate simultaneously and interact in complex ways depending on regional characteristics.

The finding that achieving near-equity requires federal cost shares of 85\% or higher carries significant budget implications. Under current HMGP funding levels, comprehensive national buyout programs remain infeasible. However, our analysis of equity-weighted mechanisms demonstrates that thoughtful policy design can achieve equivalent equity outcomes at substantially lower cost by directing resources where they generate the largest distributional improvements.

\subsection{Implementation Considerations}

While our theoretical analysis identifies mechanisms for improving equity, practical implementation faces several challenges that merit discussion. The equity-weighted cost-sharing mechanism requires political consensus on progressivity parameters and may face opposition from higher-income jurisdictions that would receive reduced federal support. Experience with other progressive federal programs suggests that framing matters: presenting equity weights as ``vulnerability adjustments'' or ``capacity-based assistance'' may prove more politically palatable than explicit income-based formulas~\cite{banzhaf2019environmental}.

Income-tiered subsidies require reliable means-testing, which creates both administrative burden and potential for gaming. The HMGP already collects income documentation for program eligibility, providing an existing administrative infrastructure. However, verification delays could slow program implementation, potentially undermining the urgency benefits of post-disaster buyouts~\cite{binder2016framework}.

Minimum service requirements represent perhaps the most politically feasible option, as they leverage existing federal civil rights frameworks. The Fair Housing Act provides precedent for requiring equitable service delivery across demographic groups, and environmental justice executive orders establish policy foundations for prioritizing underserved communities~\cite{bullard2005quest}. However, enforcement mechanisms remain challenging: penalties for non-compliance could reduce program participation, while weak enforcement may render requirements ineffective.

\subsection{Limitations}

Several caveats merit attention. Our model assumes rational utility-maximizing behaviour, yet households may exhibit bounded rationality, status quo bias~\cite{samuelson1988status}, or systematic flood risk misperception~\cite{meyer2017ostrich}. We assume perfect information about flood probabilities, whereas real-world uncertainty may further delay household decision-making. The model treats relocation as irreversible, abstracting from partial mitigation strategies such as elevation or flood-proofing that households might pursue as alternatives to buyouts.

We simplify local government objectives in our framework. Actual decision-making involves electoral cycles that create short-term horizons, interest-group pressure from both property owners and developers, and administrative capacity constraints that vary substantially across jurisdictions~\cite{greer2017political}. Our tax-base preservation parameter $\lambda_j$ captures one dimension of these political economy considerations but cannot represent the full complexity of local governance.

Finally, we do not model general equilibrium effects in housing markets. Large-scale buyout programs could affect property values in both source communities (through reduced housing supply) and destination communities (through increased demand). These feedback effects could create new equity concerns not captured in our partial equilibrium framework.

\subsection{Conclusion}

Despite these limitations, our framework provides a rigorous theoretical foundation for understanding documented disparities in buyout program implementation. The three mechanisms we identify suggest distinct policy responses: addressing the discount rate channel through income-tiered subsidies, addressing the tax-base channel through federal incentives for equitable targeting, and addressing the participation channel through progressive cost-sharing formulas.

As climate change intensifies flood risk and managed retreat becomes increasingly necessary, ensuring equitable access to buyout programs represents both a moral imperative and a practical necessity for sustainable climate adaptation. Our analysis demonstrates that equity and efficiency need not conflict---carefully designed mechanisms can improve distributional outcomes while controlling costs.

\section{Methods}
\label{sec:methods}

\subsection{Model Specification}

We model a three-level stochastic Stackelberg game with the federal government as leader, $J$ local governments as intermediate followers, and households as ultimate followers. Climate uncertainty is represented by scenarios $\omega \in \Omega$ with probabilities $\pi(\omega)$ corresponding to Representative Concentration Pathways.

The federal government chooses cost-sharing ratio $\alpha \in [\underline{\alpha}, \bar{\alpha}]$ and maximum subsidy $\bar{S}$ to minimize expected social costs subject to budget and equity constraints:
\begin{equation}
\min_{\alpha, \bar{S}} \E_\omega\left[\sum_{j} \alpha S_j^* N_j^{\text{rel}} + C_j^{\text{Fed}}\right] \quad \text{s.t.} \quad \text{RRG} \geq \theta
\end{equation}
where $S_j^*$ and $N_j^{\text{rel}}$ denote equilibrium subsidy and relocations in jurisdiction $j$.

Local government $j$ minimizes:
\begin{equation}
(1-\alpha) S_j N_j^{\text{rel}}(S_j) + C_j^{\text{admin}} + \E_\omega[C_j^{\text{loc}}(S_j)] + \lambda_j \Delta TB_j(S_j)
\end{equation}
subject to $S_j \leq \bar{S}$, where $\Delta TB_j$ represents tax-base loss from buyouts.

Household $h$ chooses relocation year $k_h$ to minimize:
\begin{equation}
\frac{M_h - S_j}{(1+r_h)^{k_h}} + \sum_{t=0}^{k_h-1} \frac{\E_\omega[\E[D_{h,t}^{(\omega)}]]}{(1+r_h)^t}
\end{equation}
where $M_h$ denotes relocation cost and $D_{h,t}^{(\omega)}$ denotes flood damage.

\subsection{Climate Scenarios}

We model three climate scenarios based on IPCC Representative Concentration Pathways~\cite{ipcc2019srocc}. RCP 2.6 (probability 0.2) projects 0.4m sea-level rise by 2100; RCP 4.5 (probability 0.5) projects 0.6m; RCP 8.5 (probability 0.3) projects 1.0m. Annual flood probabilities follow a generalized extreme value distribution with location parameter adjusted for sea-level rise following Buchanan et al.~\cite{buchanan2016allowances}.

\subsection{Damage Estimation}

Expected annual flood damage combines flood probability with depth-damage functions following US Army Corps of Engineers methodology~\cite{usace1992depth}. Structural damage ratios range from 15\% at 0.3m flood depth to 70\% at depths exceeding 2.4m, with linear interpolation between thresholds. This approach is standard in flood risk assessment~\cite{wing2018estimates} and allows comparison with existing benefit-cost analyses of buyout programs.

\subsection{Household Population}

We generate synthetic household populations for each region calibrated to observed distributions of property values, income, elevation, and flood zone location. Property values follow truncated lognormal distributions with region-specific parameters derived from PropertyShark and Zillow data. Low-income households are assigned discount rate $r^{\text{high}} = 18\%$; high-income households receive $r^{\text{low}} = 12\%$, following Warner and Pleeter~\cite{warner2001personal}.

\subsection{Solution Algorithm}

We solve the three-level game through backward induction on a discrete grid. For each federal policy $(\alpha, \bar{S})$, we compute local government best responses by numerical optimization, then aggregate household relocation decisions to obtain equilibrium outcomes. Computational details appear in Appendix~\ref{app:algorithm}.

\subsection{Data Availability}

Simulation code and parameter files are available at \url{https://github.com/yzhou364/Equity-Implications-of-Federal-Local-Cost-sharing-in-Flood-Buyouts}. Regional population parameters derive from publicly available sources: FEMA flood maps, Census Bureau income data, PropertyShark/Zillow property values, and NOAA sea-level rise projections.


\section*{Acknowledgements}
The author thanks Vicki M.\ Bier for valuable discussions and feedback on earlier versions of this work.

\section*{Author Contributions}
Y.Z.\ conceived and designed the study, developed the theoretical model, conducted the numerical analysis, and wrote the manuscript.

\section*{Competing Interests}
The authors declare no competing interests.



\appendix

\section{Mathematical Proofs}
\label{app:proofs}

\subsection{Household Optimal Relocation Timing}

\begin{theorem}[Optimal Relocation Year]
For a household with discount rate $r_h$, house value $V_h$, relocation cost $M_h$, and subsidy offer $S$, the optimal relocation year $k_h^*$ is the smallest integer $k$ such that:
\begin{equation}
\E_\omega[D_{h,k}^{(\omega)}] \geq (M_h - S) \cdot r_h
\end{equation}
where $D_{h,k}^{(\omega)}$ is the expected flood damage in year $k$ under climate scenario $\omega$.
\end{theorem}

\begin{proof}
The household minimizes total discounted costs:
\[
C_h(k) = \frac{M_h - S}{(1+r_h)^k} + \sum_{t=0}^{k-1} \frac{\E_\omega[D_{h,t}^{(\omega)}]}{(1+r_h)^t}
\]

The first-order condition $\frac{\partial C_h}{\partial k} = 0$ yields:
\[
-\frac{(M_h - S) \ln(1+r_h)}{(1+r_h)^k} + \frac{\E_\omega[D_{h,k-1}^{(\omega)}]}{(1+r_h)^{k-1}} = 0
\]

Rearranging and using the approximation $\ln(1+r_h) \approx r_h$ for small $r_h$:
\[
\E_\omega[D_{h,k-1}^{(\omega)}] = (M_h - S) \cdot r_h
\]

The optimal year is when expected annual damage first exceeds the annualized net relocation cost.
\end{proof}

\subsection{Local Government Best Response}

\begin{theorem}[Local Government Optimal Subsidy]
Given federal cost-share $\alpha$ and maximum subsidy $\bar{S}$, local government $j$'s optimal subsidy satisfies:
\begin{equation}
S_j^* = \min\left\{\bar{S}, \frac{(1-\alpha) + \lambda_j \tau_j}{\frac{\partial N_j^{\text{rel}}}{\partial S_j} \cdot \left[(1-\alpha) + \lambda_j \tau_j V_j^{\text{avg}}\right] - \frac{\partial \E[C_j^{\text{loc}}]}{\partial S_j}}\right\}
\end{equation}
where $\tau_j$ is the property tax rate and $V_j^{\text{avg}}$ is average property value.
\end{theorem}

\begin{proof}
The local government minimizes:
\[
L_j(S_j) = (1-\alpha) S_j N_j^{\text{rel}}(S_j) + C_j^{\text{admin}} + \E_\omega[C_j^{\text{loc}}(S_j)] + \lambda_j \Delta TB_j(S_j)
\]

Taking the first-order condition and solving yields the interior solution. The constraint $S_j \leq \bar{S}$ binds when the unconstrained optimum exceeds $\bar{S}$.
\end{proof}

\subsection{Equilibrium Existence}

\begin{theorem}[Existence of Subgame Perfect Equilibrium]
The three-level Stackelberg game admits a subgame perfect Nash equilibrium for any federal policy $(\alpha, \bar{S})$ satisfying $\alpha \in [0,1]$ and $\bar{S} > 0$.
\end{theorem}

\begin{proof}
We establish existence through backward induction:

\textit{Stage 3 (Households):} Given any subsidy vector $(S_1, \ldots, S_J)$, each household's optimization problem is well-defined with a unique solution $k_h^*(S_j)$ by Theorem 1.

\textit{Stage 2 (Local Governments):} Given $(\alpha, \bar{S})$, each local government's optimization is a continuous function over the compact set $[0, \bar{S}]$. By the Weierstrass theorem, an optimum exists.

\textit{Stage 1 (Federal Government):} The federal objective function is continuous in $(\alpha, \bar{S})$ over the compact policy space. The constraint set is closed. By standard arguments, an optimum exists.
\end{proof}

\section{Solution Algorithm}
\label{app:algorithm}

\begin{algorithm}[H]
\caption{Backward Induction Solution for Three-Level Stackelberg Game}
\begin{algorithmic}[1]
\Require Federal policy $(\alpha, \bar{S})$, household parameters $\{(V_h, r_h, \ell_h)\}$, climate scenarios $\Omega$
\Ensure Equilibrium outcomes: subsidies $\{S_j^*\}$, relocations $\{N_j^{\text{rel}}\}$, RRG
\Statex
\Statex \textbf{Stage 3: Household Decisions}
\For{each household $h = 1, \ldots, H$}
    \For{each subsidy level $S \in \{0, \Delta S, 2\Delta S, \ldots, \bar{S}\}$}
        \State Compute optimal relocation year $k_h^*(S)$ using Theorem 1
    \EndFor
\EndFor
\Statex
\Statex \textbf{Stage 2: Local Government Optimization}
\For{each jurisdiction $j = 1, \ldots, J$}
    \State $S_j^{\text{best}} \gets 0$; \quad $C_j^{\text{best}} \gets \infty$
    \For{each $S_j \in \{0, \Delta S, \ldots, \bar{S}\}$}
        \State $N_j^{\text{rel}}(S_j) \gets \sum_{h \in j} \mathbf{1}[k_h^*(S_j) \leq T]$
        \State Compute total cost $C_j(S_j)$ per Equation (2)
        \If{$C_j(S_j) < C_j^{\text{best}}$}
            \State $S_j^{\text{best}} \gets S_j$; \quad $C_j^{\text{best}} \gets C_j(S_j)$
        \EndIf
    \EndFor
\EndFor
\Statex
\Statex \textbf{Stage 1: Aggregate Outcomes}
\State Compute total federal cost: $C^{\text{Fed}} \gets \sum_j \alpha \cdot S_j^{\text{best}} \cdot N_j^{\text{rel}}(S_j^{\text{best}})$
\State Compute relocations by income group: $N^{\text{low}}, N^{\text{high}}$
\State Compute relocation ratio gap: $\text{RRG} \gets (N^{\text{low}}/H^{\text{low}}) / (N^{\text{high}}/H^{\text{high}})$
\State \Return $\{S_j^{\text{best}}\}, \{N_j^{\text{rel}}\}, C^{\text{Fed}}, \text{RRG}$
\end{algorithmic}
\end{algorithm}

\section{Regional Parameter Calibration}
\label{app:parameters}

\begin{table}[h]
\centering
\caption{Regional population parameters}
\label{tab:parameters}
\begin{tabular}{lcccccc}
\toprule
Region & Households & Mean Value & Low-Income & Flood Exp. & Mean Elev. \\
 & (N) & (\$000) & Share (\%) & (\%) & (m) \\
\midrule
Brooklyn & 5,200 & 650 & 45 & 35 & 2.0 \\
Staten Island & 2,800 & 580 & 52 & 42 & 1.8 \\
Queens & 4,100 & 520 & 58 & 65 & 1.5 \\
Lower Manhattan & 1,850 & 1,850 & 15 & 45 & 2.0 \\
Houston & 6,500 & 285 & 55 & 38 & 9.0 \\
New Orleans & 4,200 & 265 & 62 & 75 & 0.5 \\
Miami-Dade & 5,100 & 485 & 48 & 52 & 1.5 \\
Charleston & 2,350 & 425 & 42 & 55 & 2.5 \\
Norfolk & 2,393 & 245 & 50 & 48 & 1.5 \\
\midrule
\textbf{Total} & \textbf{34,493} & -- & -- & -- & -- \\
\bottomrule
\end{tabular}
\end{table}

Data sources: Property values from PropertyShark (NYC) and Zillow (other regions); income distributions from US Census Bureau American Community Survey; flood exposure from FEMA National Flood Hazard Layer; elevation from USGS National Elevation Dataset.


\clearpage

\begin{figure}[p]
\centering
\includegraphics[width=0.85\textwidth]{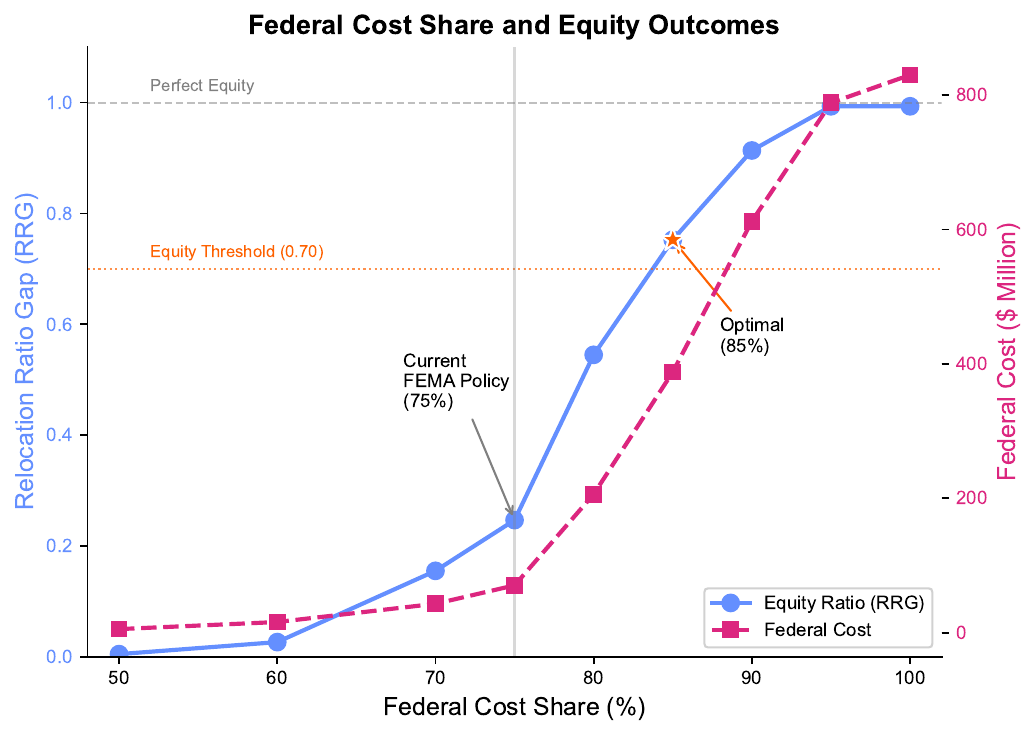}
\caption{\textbf{Equity ratio and federal cost as functions of federal cost-share ratio.} The relocation ratio gap (RRG, blue) measures the ratio of low-income to high-income relocation rates; values below 1 indicate low-income households are underserved. Under current FEMA policy ($\alpha = 0.75$, dashed vertical line), RRG equals 0.26. Achieving near-equity (RRG $\approx$ 0.70) requires $\alpha \geq 0.85$. Federal cost (red) increases nonlinearly with $\alpha$, reflecting both higher per-household subsidies and expanded program participation.}
\label{fig:equity_alpha}
\end{figure}

\begin{figure}[p]
\centering
\includegraphics[width=0.9\textwidth]{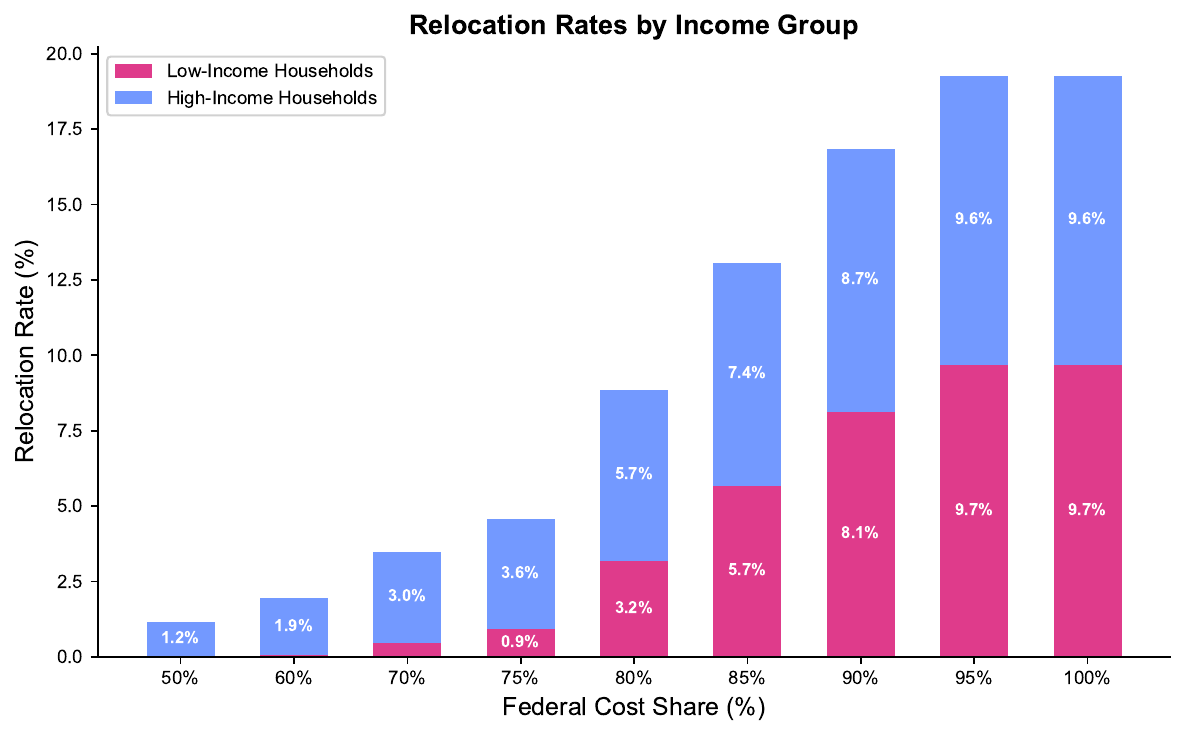}
\caption{\textbf{Relocations by income group across federal cost-share ratios.} Low-income households (dark bars) show minimal participation until $\alpha$ exceeds 0.80, reflecting the combined effects of discount rate, tax-base preservation, and participation channels. High-income households (light bars) participate at higher rates across all policy levels due to lower discount rates and residence in jurisdictions with lower participation thresholds.}
\label{fig:income}
\end{figure}

\begin{figure}[p]
\centering
\includegraphics[width=\textwidth]{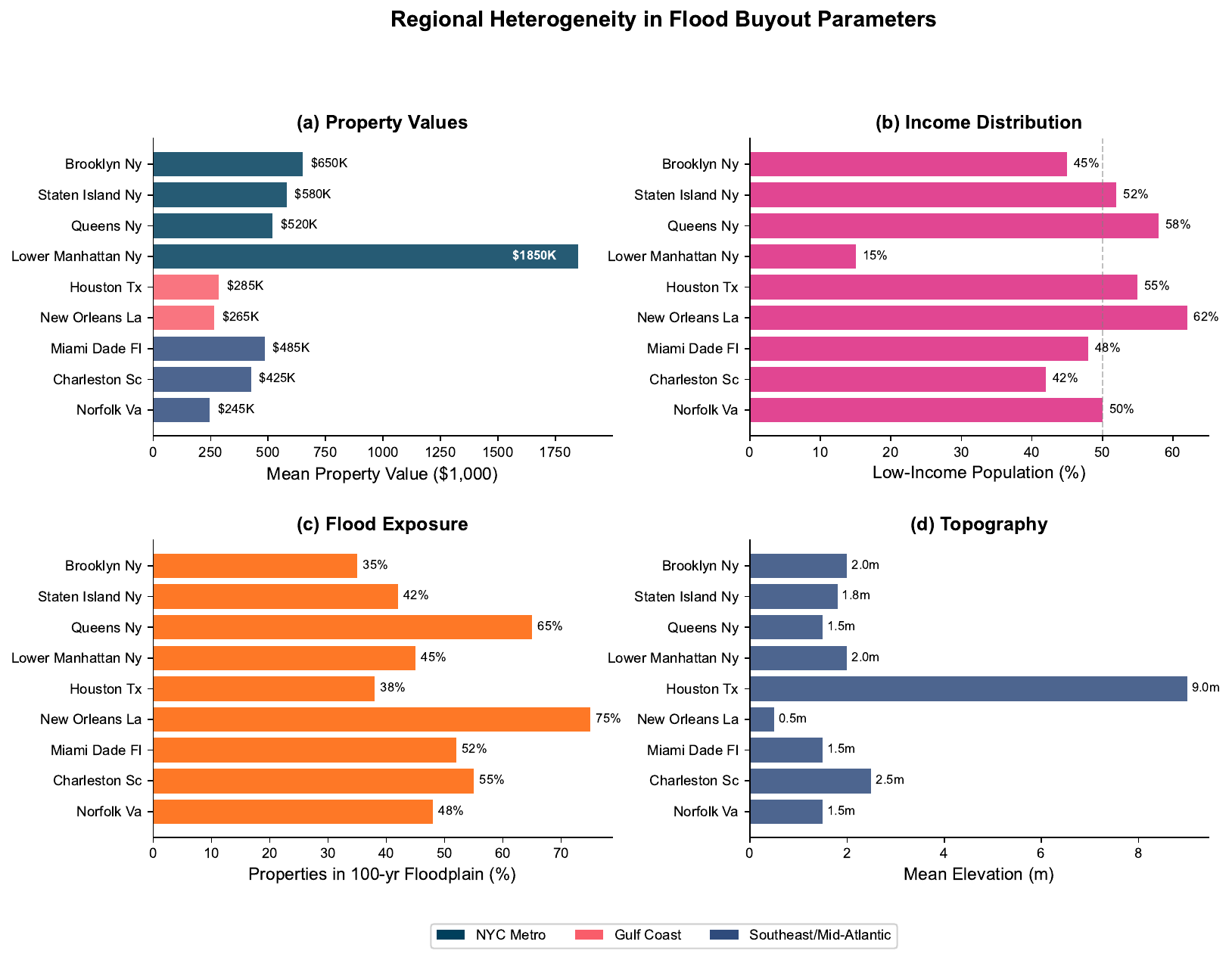}
\caption{\textbf{Regional heterogeneity across study areas.} \textbf{a}, Mean property values range from \$245,000 (Norfolk) to \$1.85 million (Lower Manhattan). \textbf{b}, Low-income population shares range from 15\% (Lower Manhattan) to 62\% (New Orleans). \textbf{c}, Flood exposure (percentage in 100-year floodplain) ranges from 35\% (Brooklyn) to 75\% (New Orleans). \textbf{d}, Mean elevation ranges from 0.5m (New Orleans, partially below sea level) to 9.0m (Houston). This heterogeneity drives substantial variation in how federal policy translates to local outcomes.}
\label{fig:regional}
\end{figure}

\begin{figure}[p]
\centering
\includegraphics[width=\textwidth]{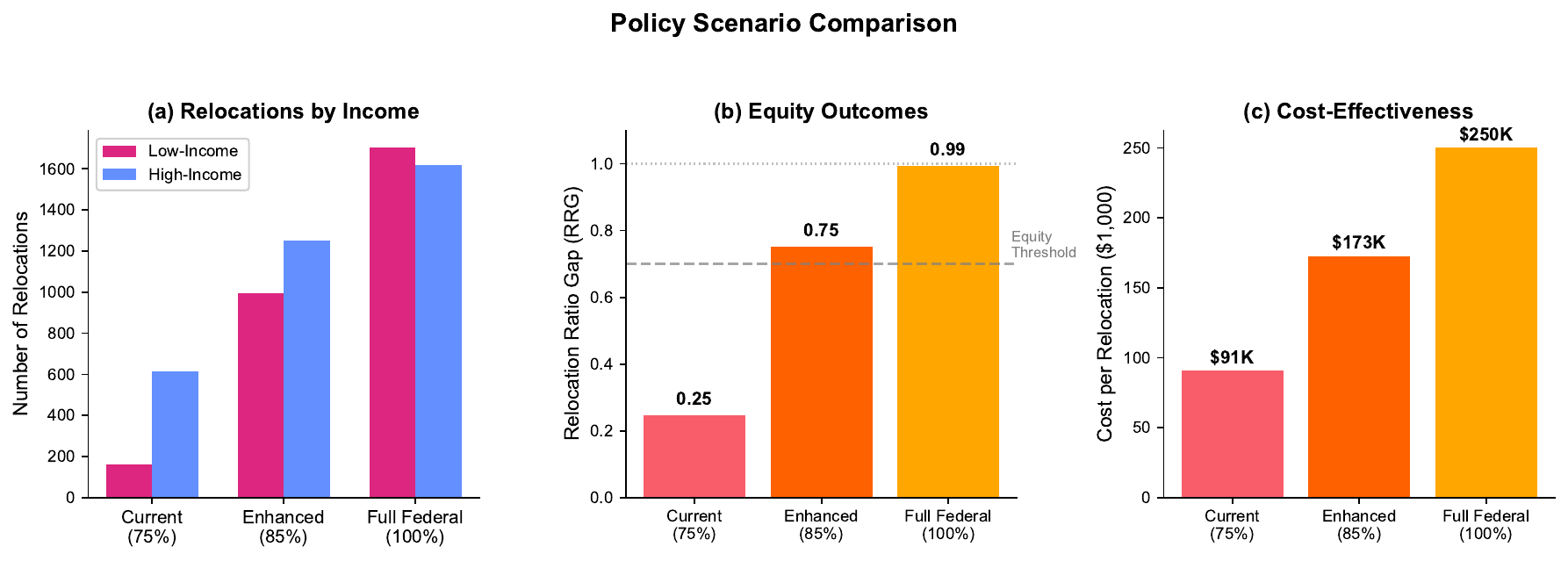}
\caption{\textbf{Comparison of policy alternatives.} \textbf{a}, Current FEMA 75/25 policy achieves low equity (RRG = 0.26) at moderate cost (\$82 million). \textbf{b}, Uniform high federal share ($\alpha = 0.90$) achieves high equity (RRG = 0.87) but at substantial cost (\$648 million). \textbf{c}, Equity-weighted cost sharing achieves comparable equity (RRG = 0.78) at 25\% lower cost (\$420 million) by directing resources to lower-income communities where marginal improvements in equity are largest.}
\label{fig:policy}
\end{figure}

\begin{figure}[p]
\centering
\includegraphics[width=0.75\textwidth]{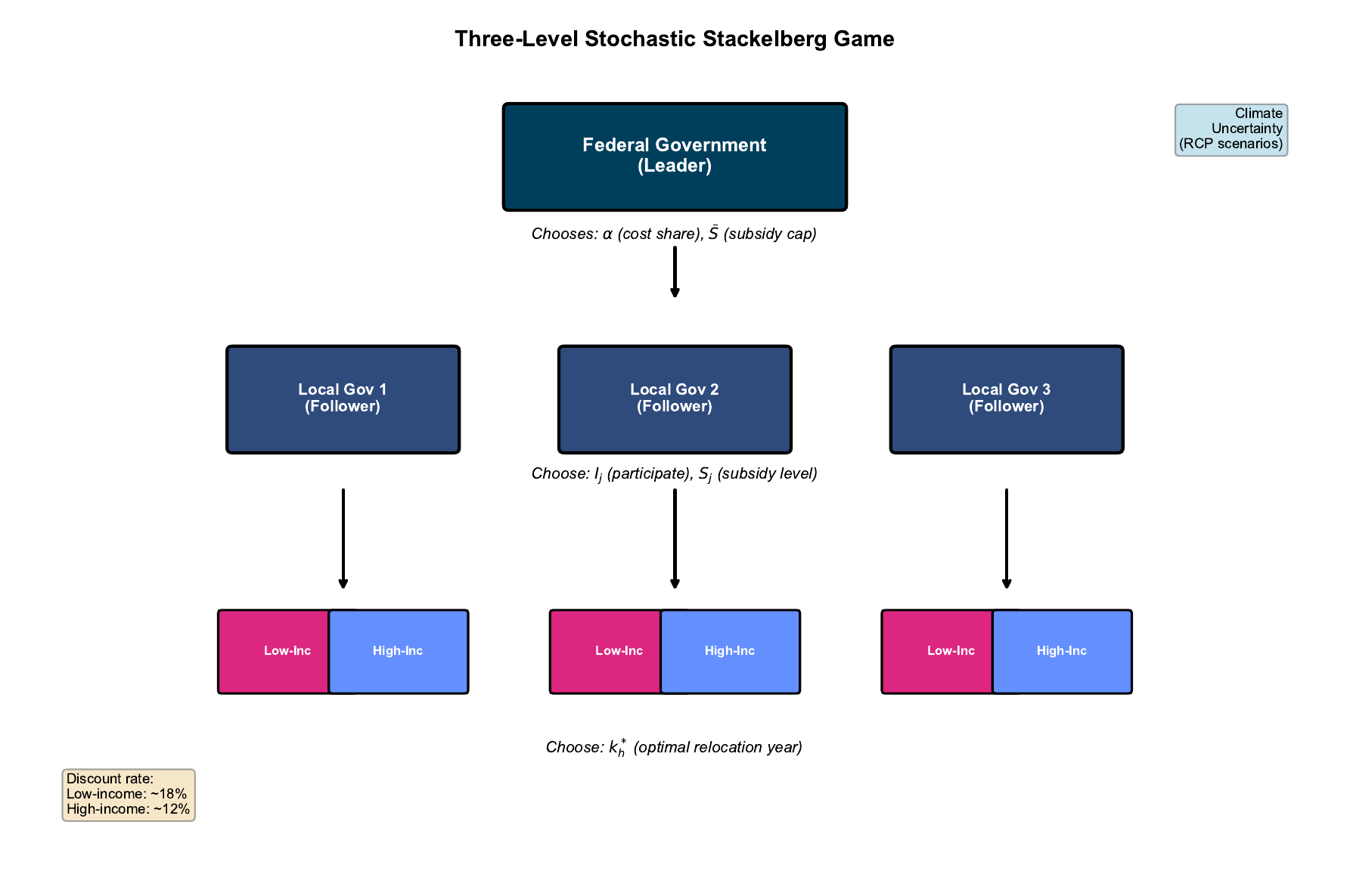}
\caption{\textbf{Three-level Stackelberg game structure.} The federal government (top) sets cost-sharing parameters $(\alpha, \bar{S})$ anticipating local government and household responses. Local governments (middle) choose participation and subsidy levels given federal policy. Households (bottom) make relocation decisions given offered subsidies. Backward induction from household decisions through local government responses yields subgame perfect equilibrium.}
\label{fig:mechanisms}
\end{figure}


\begin{table}[p]
\centering
\caption{\textbf{Model validation against historical buyout programs.} Predicted relocation rates from calibrated model compared with observed rates from four major US buyout programs. Prediction errors range from $-13$\% to $+16$\%, with slight overprediction for post-disaster programs (Staten Island, Harris County) and underprediction for proactive programs (Blue Acres, Charlotte-Mecklenburg).}
\label{tab:validation}
\begin{tabular}{lcccl}
\toprule
Program & Observed & Predicted & Error & Context \\
\midrule
Staten Island (Sandy) & 8.2\% & 9.1\% & +11\% & Post-disaster \\
Harris County (Harvey) & 4.5\% & 5.2\% & +16\% & Post-disaster \\
NJ Blue Acres & 6.8\% & 5.9\% & $-13$\% & Proactive \\
Charlotte-Mecklenburg & 12.1\% & 10.8\% & $-11$\% & Proactive \\
\bottomrule
\end{tabular}
\end{table}

\end{document}